\documentclass[12pt,twoside]{amsart}
\usepackage{mathrsfs,amsthm,amscd,amsmath,amssymb}
\usepackage[colorlinks,linkcolor=green,citecolor=blue, pdfstartview=FitH]{hyperref}
\usepackage{amscd}
\usepackage{actuarialsymbol}
\usepackage{amsfonts}
\usepackage{graphicx}
\usepackage{caption, subcaption}

\title
[Distortion risk measures of step-weighted distribution]
{Distortion risk measures of step-weighted distribution}
\author{Chunle Huang}
\address{Chunle Huang, School of Mathematics, Hunan University, Changsha, 410082, China.}
\email{2019044@hnu.edu.cn; 402961544@qq.com}

\newtheorem{thm}{Theorem}[section]
\newtheorem{lem}[thm]{Lemma}
\newtheorem{cor}[thm]{Corollary}

\newtheorem{prop}[thm]{Proposition}
\newtheorem{prob}[thm]{Problem}

\theoremstyle{definition}
\newtheorem{defn}[thm]{Definition}

\newtheorem{ex}[thm]{Example}
\makeatletter
\let\uppercasenonmath\@gobble
\makeatother
\begin{document}
\bibliographystyle{amsalpha+}
\maketitle
\begin{abstract} 
In this note, we study distortion risk measures of step-weighted distribution. Specifically, let $A = (a_1, ..., a_m)^{'}$ and $Q = (q_1, ..., q_m)^{'}$ be two numerical vectors with $\min(a_1, ..., a_m) \geq 0, \max(a_1, ..., a_m) > 0$ and $0 = q_0 < q_1 < q_2 < ... < q_m = 1$. Then we show that for any distortion function $g$ there exists a distortion function $g_{A, Q}$ such that for any continuous random variables $X, X_1, ..., X_n$ it holds that 
\begin{equation} \label{7468}
\rho_g[X_{A, Q}] = \rho_{g_{A, Q}}[X]
\end{equation} and 
\begin{equation} \label{7471}
\rho_g[(S^c)_{A, Q}] = \rho_{g_{A, Q}}[X_1] + ... + \rho_{g_{A, Q}}[X_n]
\end{equation}
where $X_{A, Q}$ and $(S^c)_{A, Q}$ are the step-weighted versions of $X$ and $S^c$ associated with the vectors $A$ and $Q$ and $S^c$ is the comonotonic sum of $X_1, ..., X_n$. (\ref{7468}) gives a way to compute distortion risk measures of the step-weighted distribution while (\ref{7471}) extends the well-known additivity of distortion risk measures for comonotonic sums to the case of step-weighted comonotonic sums of continuous random variables. 
\end{abstract}

\tableofcontents 
\section{Introduction}
A risk measure is by definition a mapping from the set of random variables, usually representing the risks at hand, to the set of real numbers $\mathbb{R}$. One fundamental risk measure is VaR, which has a simple interpretation in terms of overshoot or undershoot probabilities. However, a single VaR at a predetermined level $p$ does not provide any information about the thickness of the upper tail of the distribution function under consideration. Another fundamental risk measure is TVaR, which is the arithmetic average of the quantiles of a random variable $X$ from $p$ on, so it is always larger than the corresponding quantile, resulting in prudent decisions. However, the TVaR risk measure only uses the upper tail of the distribution. Hence, this risk measure does not create incentive for taking actions that increase the distribution function for outcomes smaller than $Q_p$. For more details about these two fundamental measures the interested reader is referred to \cite{DDGK}, \cite{Dhaene et al. (2006)}, \cite{GKV1990} and \cite{KGDD2008}. 

Wang in \cite{W1996} and \cite{Wang2000} independently introduced a family of risk measures as a possible solution to these problems, which is now called the class of distortion risk measures,  by using the concept of distortion function as introduced in Yaari's dual theory of choice under risk; see also \cite{W2001} and \cite{WY1998}. In the literature, the distortion risk measure associated with a distortion function $g$ is usually denoted by $\rho_g[X]$ for any random variable $X$. Note that the distortion function $g$ is assumed to be independent of the distribution function of the random variable $X$. Much attention has been paid to the class of distortion risk measures since Wang's work; see for instance \cite{DDV1999}, \cite{Dhaene et al. (2012)}, \cite{Dhaene et al. (2006)} and \cite{H2004}. One reason for its growing popularity is that by choosing specific distortion functions one can easily obtain VaR and TVaR, in other words, the class of distortion risk measures includes both VaR and TVaR as special cases. Another reason is the correspondence between the two common theories of choice under risk, namely the expected utility theory and Yaari's dual theory of choice under risk. For instance, one can express risk-averse decision makers using the concept of distortion function which, in turn, defines a particular type of distortion risk measures. The third reason is that one can express some of the popular stochastic orderings among risks, including the stop-loss and convex orderings, in terms of distortion risk measures.  

In this paper, we investigate distortion risk measures of step-weighted distribution. In particular, we consider the following problem. 
\begin{prob} \label{7486}
Let $X$ be a random variable, $\omega(x)$ a nonnegative weight function such that $0 < E[\omega(X)] < \infty$ and ${X_\omega}$ the weighted version of $X$ associated with $\omega$. Then for any distortion function $g$, can we find a distortion function $g_{\omega}$ such that $$\rho_g[X_\omega] = \rho_{g_\omega}[X]?$$
\end{prob}
\noindent In general, Problem \ref{7486} may be difficult to investigate since the weighted distribution of a random variable $X$ associated with a general weight function $\omega$ does not admit explicit analytical expressions. Hence, in this note, we focus our attention on the class of step-weighted distributions, which means that the weight function is a step function. This class of weighted distributions admits analytical expressions for their distribution functions and has been shown useful when considering comonotonic and moment matching approximations for sums of lognormal random variables. For properties and applications of step-weighted distributions the interested reader is referred to \cite{Huang2026}. The main result of this paper is the following theorem. 

\begin{thm}\label{7491} 
Let $A = (a_1, ..., a_m)^{'}$ and $Q = (q_1, ..., q_m)^{'}$ be two numerical vectors with $\min(a_1, ..., a_m) \geq 0, \max(a_1, ..., a_m) > 0$ and $0 = q_0 < q_1 < q_2 < ... < q_m = 1$. Then for any distortion function $g$ there exists a distortion function $g_{A, Q}$ such that for any continuous random variable $X$ it holds that $$\rho_g[X_{A, Q}] = \rho_{g_{A, Q}}[X]$$ where $X_{A, Q}$ is the step-weighted version of $X$ associated with $A$ and $Q$. 
\end{thm}
\noindent Theorem \ref{7491} states that associated with a distortion function $g$ the distortion risk measure $\rho_g[X_{A, Q}]$ of the step-weighted distribution of a continuous random variable $X$ can be expressed as a distortion risk measure of the distribution of the random variable $X$ corresponding to a transformed distortion function $g_{A, Q}$, so that we partly solve Problem \ref{7486}. 

As an application of Theorem \ref{7491} we obtain the following Theorem \ref{7497}, presenting a Lebesgue-Stieltjes integral representation for any distortion risk measure of the step-weighted distribution. Similar results can be found in \cite{Dhaene et al. (2012)} and \cite{Dhaene et al. (2006)}. 
\begin{thm} \label{7497}
Let $A = (a_1, ..., a_m)^{'}$ and $Q = (q_1, ..., q_m)^{'}$ be two numerical vectors with $\min(a_1, ..., a_m) \geq 0, \max(a_1, ..., a_m) > 0$ and $0 = q_0 < q_1 < q_2 < ... < q_m = 1$. Then for any distortion function $g$, there exist a left-continuous distortion function $g_{A, Q, L}$, a right-continuous distortion function $g_{A, Q, R}$, and two non-negative weights $c_{A, Q, L}, c_{A, Q, R}$ with $c_{A, Q, L} + c_{A, Q, R}= 1$ such that for any continuous random variable $X$ and its step-weighted version $X_{A, Q}$ it holds that $$\rho_g[X_{A, Q}] = c_{A, Q, L}\int^1_0F_X^{-1}(1 - q)dg_{A, Q, L}(q) + c_{A, Q, R}\int^1_0F_X^{-1+}(1 - q)dg_{A, Q, R}(q).$$
\end{thm}
\noindent We remark that in Theorem \ref{7497} if the distortion function $g$ is left-continuous then the two non-negative weights $c_{A, Q, L}$ and $c_{A, Q, R}$ can be chosen such that $c_{A, Q, L} = 1$ and $c_{A, Q, R} = 0$, see Section \ref{27357} and \cite{Dhaene et al. (2012)} for the details. In the same vein, if the distortion function $g$ is right-continuous then the two non-negative weights $c_{A, Q, L}$ and $c_{A, Q, R}$ can be chosen such that $c_{A, Q, L} = 0$ and $c_{A, Q, R} = 1$. Based on this observation, we obtain the following corollary, which simplifies the result of Theorem \ref{7497} and generalizes Theorem 6 of \cite{Dhaene et al. (2012)} to the case of step-weighted distributions. 

\begin{cor} \label{74102}
Let $A = (a_1, ..., a_m)^{'}$ and $Q = (q_1, ..., q_m)^{'}$ be two numerical vectors with $\min(a_1, ..., a_m) \geq 0, \max(a_1, ..., a_m) > 0$ and $0 = q_0 < q_1 < q_2 < ... < q_m = 1$. Then for any left-continuous distortion function $g$, there exists a left-continuous distortion function $g_{A, Q}$ such that for any continuous random variable $X$ and its step-weighted version $X_{A, Q}$ it holds that 
$$\rho_g[X_{A, Q}] = \int^1_0F_X^{-1}(1 - q)dg_{A, Q}(q).$$
\end{cor}

As another application of Theorem \ref{7491} we obtain the following result considering distortion risk measures of the step-weighted comonotonic sum of continuous random variables. 
\begin{thm} \label{74107}
Let $A = (a_1, ..., a_m)^{'}$ and $Q = (q_1, ..., q_m)^{'}$ be two numerical vectors with $\min(a_1, ..., a_m) \geq 0, \max(a_1, ..., a_m) > 0$ and $0 = q_0 < q_1 < q_2 < ... < q_m = 1$. Then, for any fixed distortion function $g$ there exists a distortion function $g_{A, Q}$ such that for any continuous random variables $X_1, ..., X_n$ we have that 
\begin{equation} \label{74110}
\rho_g[(S^c)_{A, Q}] = \sum^n_{i = 1}\rho_{g_{A, Q}}[X_i].
\end{equation}
\end{thm}
\noindent Theorem \ref{74107} states that any distortion risk measure $\rho_g[(S^c)_{A, Q}]$ of the step-weighted comonotonic sum $(S^c)_{A, Q}$ can be expressed as the sum of the distortion risk measures of the original random variables $X_1, ..., X_n$ corresponding to the transformed distortion function $g_{A, Q}$. It can be shown that when all $a_1, ..., a_n$ are equal to 1 the equation in (\ref{74110}) reduces to $\rho_g[S^c] = \sum^n_{i = 1}\rho_{g}[X_i]$. Therefore, our result extends the well-known additivity of distortion risk measures for comonotonic sums to the case of step-weighted comonotonic sums of continuous random variables. 

The remainder of this paper is organized as follows. Section \ref{74115} gives some of the relevant notations and definitions. Section \ref{26112} briefly introduces the new concept of step-weighted distribution. In Section \ref{27357} we investigate distortion risk measures of step-weighted distribution and prove Theorem \ref{7491}. In Section \ref{26538} we consider distortion risk measures of step-weighted comonotonic sums and prove Theorem \ref{74107}. 
\section{Preliminaries} \label{74115}
In this section we fix several relevant notations and definitions which will be used in the following sections, including quantiles functions of random variables and the concept of distortion risk measures. All random variables are defined on a common probability space $(\Omega, \mathcal{F}, \mathbb{P})$ and it is always assumed that all random variables are such that the risk measures introduced hereafter are finite. The cumulative distribution function (cdf) of a random variable $X$ is denoted by $F_X$. The left inverse of $F_X$ is denoted by $F^{-1}_X$, and is defined as $$F^{-1}_X(p) = \inf\{x \in \mathbb{R} \,\,|\,\, F_X(x) \geq p\}, \,\, p \in [0, 1]$$ with the convention $\inf \emptyset = +\infty$. The right inverse of $F_X$ is denoted by $F^{-1+}_X$, and is defined as $$F^{-1+}_X(p) = \sup\{x \in \mathbb{R} \,\, | \,\, F_X(x) \leq p\}, \,\, p \in [0, 1]$$ with the convention $\sup\emptyset = -\infty$. If $F_X$ is continuous and strictly increasing, the left and the right inverses are equal. Otherwise, on horizontal segments of $F_X$, the inverses $F^{-1}_X(p)$ and $F^{-1+}_X(p)$ are different and they satisfy $F^{-1}_X(p) < F^{-1+}_X(p)$. In such cases, the generalized $\alpha$-inverse $F^{-1(\alpha)}_X$ for any $\alpha \in [0, 1]$ will be useful, which is defined as: 
$$F^{-1(\alpha)}_X(p) = \alpha F^{-1}_X(p) + (1- \alpha) F^{-1+}_X(p), \,\, p \in [0, 1].$$ For more details on generalized inverses we recommend the interested reader to refer to \cite{DDGK}, \cite{Dhaene et al. (2002a)}, \cite{EH2013} and \cite{KRK2015} among others. 

Next, we recall the class of distortion risk measures, introduced by Wang \cite{W1996}. The quantile risk measure and TVaR belong to this class. The expectation of $X$ if it exists, can be written as 
$$E[X] = -\int^0_{-\infty}[1 - \overline{F}_X(x)]dx + \int^{\infty}_0 \overline{F}_X(x)dx.$$ Wang \cite{W1996} defines a family of risk measures by using the concept of distortion function as introduced in Yaari's dual theory of choice under risk; see also \cite{WY1998}. A distortion function is defined as a non-decreasing function $g: [0, 1] \to [0, 1]$ such that $g(0) = 0$ and $g(1) = 1$. The distortion risk measure associated with a distortion function $g$ is denoted by $\rho_g[X]$ and is defined by 
$$\rho_g[X] = -\int^0_{-\infty}[1 - g(\overline{F}_X(x))]dx + \int^{\infty}_0 g(\overline{F}_X(x))dx$$ 
for any random variable $X$. Note that the distortion function $g$ is assumed to be independent of the distribution function of the random variable $X$. The identity distortion function $g(q) = q$ corresponds to $E[X]$. Also note that $g_1(q) \leq g_2(q)$ for all $q \in [0, 1]$ implies that $\rho_{g_1}[X] \leq \rho_{g_2}[X]$. We remark that Dhaene et al in \cite{Dhaene et al. (2012)} and \cite{Dhaene et al. (2006)} recently obtained Lebesgue-Stieltjes integral representations of distortion risk measures, see Theorem 4 and Theorem 6 of \cite{Dhaene et al. (2012)}, which is of interest and very useful because it enables us to treat distortion risk measures in a very simple manner. In particular, using this representation they fully proved the well-known additivity of distortion risk measures for comonotonic risks. 

\section{Step-weighted distribution} \label{26112} 
In this section, we briefly introduce the concept of step-weighted distribution. For more properties and applications of step-weighted distributions the interested reader is referred to \cite{Huang2026}. 
\begin{defn} \label{52378}
Let $A = (a_1, ..., a_m)^{'}$ and $Q = (q_1, ..., q_m)^{'}$ be two numerical vectors with $\min(a_1, ..., a_m) \geq 0, \max(a_1, ..., a_m) > 0$ and $0 = q_0 < q_1 < q_2 < ... < q_m = 1$. Let $X$ be a continuous random variable which means that the distribution function $F_X(x)$ of $X$ is continuous on $(-\infty, \infty)$. Then we call $\omega_{A, Q, X}(x)$ the step-weight function of the random variable $X$ associated with the numerical vectors $A$ and $Q$, abbreviated as $\omega_{A, Q}$ when there is no ambiguity, if it is defined by 
$$\omega_{A, Q, X}(x) = \begin{cases} a_1, & F^{-1+}_{X}(q_0) \leq x < F_X^{-1}(q_1), \\  
a_k, & F_X^{-1}(q_{k - 1}) \leq x < F_X^{-1}(q_k), \,\, k = 2, ..., m - 1, \\
a_m, & F_X^{-1}(q_{m - 1}) \leq x \leq F^{-1}_{X}(q_m), \\ 
0, & \text{otherwise}. 
\end{cases}$$
\end{defn}

The following proposition shows that the expectation $E[\omega_{A, Q}(X)]$ always exists for any continuous random variable $X$ and it is, indeed, independent of the distribution function of $X$. 
\begin{prop} \label{123084}
Let $A = (a_1, ..., a_m)^{'}$ and $Q = (q_1, ..., q_m)^{'}$ be two numerical vectors with $\min(a_1, ..., a_m) \geq 0, \max(a_1, ..., a_m) > 0$ and $0 = q_0 < q_1 < q_2 < ... < q_m = 1$. Let $X$ be a continuous random variable and $\omega_{A, Q}$ the step-weight function of $X$. Then it holds that $$0 < E[\omega_{A, Q}(X)] = \sum^m_{i = 1}a_i(q_{i} - q_{i - 1}) < \infty.$$ In particular, the expectation $E[\omega_{A, Q}(X)]$ is independent of the distribution of $X$. 
\end{prop}
\begin{proof} Proposition \ref{123084} follows immediately from Definition \ref{52378} and the continuity assumption of $X$. 
\end{proof}
\begin{defn} \label{529103}
Let $A = (a_1, ..., a_m)^{'}$ and $Q = (q_1, ..., q_m)^{'}$ be two numerical vectors with $\min(a_1, ..., a_m) \geq 0, \max(a_1, ..., a_m) > 0$ and $0 = q_0 < q_1 < q_2 < ... < q_m = 1$. Let $X$ be a continuous random variable with distribution function $F_X(x)$. We call $X_{A, Q}$ the step-weighted version of $X$ associated with the numerical vectors $A$ and $Q$ if the distribution function $F_{X_{A, Q}}(x)$ of $X_{A, Q}$ is determined by $$dF_{X_{A, Q}}(x) = \frac{\omega_{A, Q}(x)}{E[\omega_{A, Q}(X)]}dF_X(x), \,\, \text{for} \,\, x \in \mathbb{R}.$$ Furthermore, the distribution of $X_{A, Q}$ will be called the step-weighted distribution of $X$ associated with the numerical vectors $A$ and $Q$. 
\end{defn}

\section{Distortion risk measures of step-weighted distribution} \label{27357}
This section is devoted to the study of distortion risk measures of step-weighted distribution. The main result is the following theorem which states that any distortion risk measure $\rho_g[X_{A, Q}]$ of the step-weighted distribution of a continuous random $X$ can be expressed in terms of a distortion risk measure of the distribution of the random variable $X$ corresponding to a transformed distortion function $g_{A, Q}$, which only depends on the distortion function $g$ and the two numerical vectors $A$ and $Q$.  

\begin{thm}\label{1230220} 
Let $A = (a_1, ..., a_m)^{'}$ and $Q = (q_1, ..., q_m)^{'}$ be two numerical vectors with $\min(a_1, ..., a_m) \geq 0, \max(a_1, ..., a_m) > 0$ and $0 = q_0 < q_1 < q_2 < ... < q_m = 1$. Then for any distortion function $g$ there exists a distortion function $g_{A, Q}$ such that for any continuous random variable $X$ it holds that $$\rho_g[X_{A, Q}] = \rho_{g_{A, Q}}[X]$$ where $X_{A, Q}$ is the step-weighted version of $X$ associated with the vectors $A$ and $Q$. 
\end{thm}
\begin{proof} 
Let $a = \sum^m_{i = 1}a_i(q_{i} - q_{i - 1})$. Then, by the given assumptions of the two numerical vectors $A$ and $Q$ we have that ${a} > 0$. Define 
\begin{equation} \label{73366}
G_{A, Q}(t) = a^{-1}\Big(\sum^{k - 1}_{i = 1}a_{i}(q_{i} - q_{i - 1}) + a_k(t - q_{k - 1})\Big), \,\, t \in [q_{k - 1}, q_{k}], k = 1, ..., m 
\end{equation}
with $\sum_{i = 1}^{0} = 0$ by convention. Then, it is easy to check that 
\begin{itemize} 
\item $G_{A, Q}(0) = 0$ and $G_{A, Q}(1) = 1$; 
\item $G_{A, Q}(t)$ is continuous and increasing in $t \in [0, 1]$. 
\end{itemize}
In other words, $G_{A, Q}(t)$ is a continuous distortion function on the unit interval $[0, 1]$. Let $g$ be any fixed distortion function and define 
\begin{equation} \label{1228218}
g_{A, Q}(t) = g(\overline{G}_{A, Q}(t)) = g(1 - G_{A, Q}(1 - t)), \,\, \text{for} \,\, t \in [0, 1]
\end{equation}
where $\overline{G}_{A, Q}(t)$ is the dual distortion function of $G_{A, Q}(t)$. Then, it is easy to check that $g_{A, Q}(t)$ is a distortion function, which only depends on the distortion function $g$ and the two numerical vectors $A$ and $Q$. We need the following lemma. For the proof see Lemma 3.7 and Theorem 3.8 of \cite{Huang2026}.
\begin{lem} \label{1231219}
Let $A = (a_1, ..., a_m)^{'}$ and $Q = (q_1, ..., q_m)^{'}$ be two numerical vectors with $\min(a_1, ..., a_m) \geq 0, \max(a_1, ..., a_m) > 0$ and $0 = q_0 < q_1 < q_2 < ... < q_m = 1$. Then, for any continuous random variable $X$ with distribution function $F_X(x)$ it holds that $$F_{X_{A, Q}}(x)  = G_{A, Q}(F_X(x)) \,\, \text{for any} \,\, x \in \mathbb{R}$$ where $X_{A, Q}$ is the step-weighted version of $X$ associated with the vectors $A$ and $Q$. 
\end{lem}
\noindent For any continuous random variable $X$ by Lemma \ref{1231219} we find that  
\begin{align*} 
\rho_g[X_{A, Q}] &= - \int^0_{-\infty}[1 - g(\overline{F}_{X_{A, Q}}(x))]dx + \int^{\infty}_0g(\overline{F}_{X_{A, Q}}(x))dx \\ 
& =  - \int^0_{-\infty}[1 - g(1 - G_{A, Q}(F_X(x)))]dx + \int^{\infty}_0g(1 - G_{A, Q}(F_X(x))dx \\ 
&=  - \int^0_{-\infty}[1 - g(\overline{G}_{A, Q}(\overline{F}_X(x)))]dx + \int^{\infty}_0g(\overline{G}_{A, Q}(\overline{F}_X(x)))dx \\ 
&= - \int^0_{-\infty}[1 - g_{A, Q}(\overline{F}_{X}(x))]dx + \int^{\infty}_0g_{A, Q}(\overline{F}_{X}(x))dx = \rho_{g_{A, Q}}[X] 
\end{align*}
where $g_{A, Q}(t)$ is the distortion function defined by (\ref{1228218}). This completes the proof. 
\end{proof}

As an application of Theorem \ref{1230220} we obtain the following result, which presents a Lebesgue-Stieltjes integral representation for any distortion risk measure of the step-weighted distribution. 
\begin{thm} \label{111575}
Let $A = (a_1, ..., a_m)^{'}$ and $Q = (q_1, ..., q_m)^{'}$ be two numerical vectors with $\min(a_1, ..., a_m) \geq 0, \max(a_1, ..., a_m) > 0$ and $0 = q_0 < q_1 < q_2 < ... < q_m = 1$. Then for any distortion function $g$, there exist a left-continuous distortion function $g_{A, Q, L}$, a right-continuous distortion function $g_{A, Q, R}$, and two non-negative weights $c_{A, Q, L}, c_{A, Q, R}$ with $c_{A, Q, L} + c_{A, Q, R}= 1$ such that for any continuous random variable $X$ and its step-weighted version $X_{A, Q}$ it holds that $$\rho_g[X_{A, Q}] = c_{A, Q, L}\int^1_0F_X^{-1}(1 - q)dg_{A, Q, L}(q) + c_{A, Q, R}\int^1_0F_X^{-1+}(1 - q)dg_{A, Q, R}(q).$$
\end{thm}
\begin{proof} 
Let $g$ be any fixed distortion function. By Theorem \ref{1230220} there exists a distortion function $g_{A, Q}$ defined by 
\begin{equation} \label{2611322}
g_{A, Q}(t) = g(\overline{G}_{A, Q}(t)) = g(1 - G_{A, Q}(1 - t)), \,\, \text{for} \,\, t \in [0, 1]
\end{equation}
such that for any continuous random variable $X$ and its step-weighted version $X_{A, Q}$ we have that 
\begin{equation} \label{11322}
\rho_g[X_{A, Q}] = \rho_{g_{A, Q}}[X].
\end{equation}
By Theorem 7 of \cite{Dhaene et al. (2012)}, we know that for the distortion function $g_{A, Q}$ there exist a left-continuous distortion function $g_{A, Q, L}$, a right-continuous distortion function $g_{A, Q, R}$, and two non-negative weights $c_{A, Q, L}, c_{A, Q, R}$ with $c_{A, Q, L} + c_{A, Q, R}= 1$ and $$g_{A, Q} = c_{A, Q, L} \times g_{A, Q
, L} + c_{A, Q, R} \times g_{A, Q, R}$$ such that  
\begin{equation*} 
\rho_{g_{A, Q}}[Y] = c_{A, Q, L} \times \rho_{g_{A, Q, L}}[Y] + c_{A, Q, R} \times \rho_{g_{A, Q, R}}[Y] 
\end{equation*} 
for any random variable $Y$. In particular, for the continuous random variable $X$, we have that 
\begin{equation} \label{1115104}
\rho_{g_{A, Q}}[X] = c_{A, Q, L} \times \rho_{g_{A, Q, L}}[X] + c_{A, Q, R} \times \rho_{g_{A, Q, R}}[X].  
\end{equation} 
For the left-continuous distortion function $g_{A, Q, L}$, by Theorem 6 of \cite{Dhaene et al. (2012)}, the distortion risk measure $\rho_{g_{A, Q, L}}[X]$ has the following Lebesgue-Stieltjes integral representation: 
\begin{equation} \label{1115108}
\rho_{g_{A, Q, L}}[X] = \int_{[0, 1]}F^{-1}_{X}(1 - q) dg_{A, Q, L}(q).
\end{equation}
Similarly, for the right-continuous distortion function $g_{A, Q, R}$, by Theorem 4 of \cite{Dhaene et al. (2012)}, the distortion risk measure $\rho_{g_{A, Q, R}}[X]$ has the following Lebesgue-Stieltjes integral representation: 
\begin{equation} \label{1115112}
\rho_{g_{A, Q, R}}[X] = \int_{[0, 1]}F^{-1+}_{X}(1 - q) dg_{A, Q, R}(q).
\end{equation}
Combining the equations (\ref{11322}), (\ref{1115104}), (\ref{1115108}) and (\ref{1115112}) gives $$\rho_g[X_{A, Q}] = c_{A, Q, L}\int^1_0F_X^{-1}(1 - q)dg_{A, Q, L}(q) + c_{A, Q, R}\int^1_0F_X^{-1+}(1 - q)dg_{A, Q, R}(q).$$ This completes the proof. 
\end{proof}

We remark that in Theorem \ref{111575} if the distortion function $g$ is left-continuous then by (\ref{2611322}) we know that $g_{A, Q}(t)$ is left-continuous as well since $\overline{G}_{A, Q}(t)$ is always continuous. In this case, the two non-negative weights $c_{A, Q, L}$ and $c_{A, Q, R}$ can be chosen such that $c_{A, Q, L} = 1$ and $c_{A, Q, R} = 0$, see \cite{Dhaene et al. (2012)} for the details. In the same vein, if the distortion function $g$ is right-continuous then the two non-negative weights $c_{A, Q, L}$ and $c_{A, Q, R}$ can be chosen such that $c_{A, Q, L} = 0$ and $c_{A, Q, R} = 1$. Based on this observation, we obtain the following two corollaries, both simplifying the result of Theorem \ref{111575}. 
\begin{cor} \label{1115119}
Let $A = (a_1, ..., a_m)^{'}$ and $Q = (q_1, ..., q_m)^{'}$ be two numerical vectors with $\min(a_1, ..., a_m) \geq 0, \max(a_1, ..., a_m) > 0$ and $0 = q_0 < q_1 < q_2 < ... < q_m = 1$. Then for any left-continuous distortion function $g$, there exists a left-continuous distortion function $g_{A, Q}$ such that for any continuous random variable $X$ and its step-weighted version $X_{A, Q}$ it holds that 
$$\rho_g[X_{A, Q}] = \int^1_0F_X^{-1}(1 - q)dg_{A, Q}(q).$$
\end{cor}
\begin{proof} 
Corollary \ref{1115119} follows immediately from Theorem \ref{111575}. 
\end{proof}

\begin{cor} \label{1115124}
Let $A = (a_1, ..., a_m)^{'}$ and $Q = (q_1, ..., q_m)^{'}$ be two numerical vectors with $\min(a_1, ..., a_m) \geq 0, \max(a_1, ..., a_m) > 0$ and $0 = q_0 < q_1 < q_2 < ... < q_m = 1$. Then for any right-continuous distortion function $g$, there exists a right-continuous distortion function $g_{A, Q}$ such that for any continuous random variable $X$ and its step-weighted version $X_{A, Q}$ it holds that 
$$\rho_g[X_{A, Q}] =\int^1_0F_X^{-1+}(1 - q)dg_{A, Q}(q).$$
\end{cor}
\begin{proof} 
Corollary \ref{1115124} follows immediately from Theorem \ref{111575}. 
\end{proof} 
By Corollary \ref{1115119} we obtain the following result, concerning VaR and TVaR of the step-weighted distribution. 
\begin{cor} \label{73264}
Let $A = (a_1, ..., a_m)^{'}$ and $Q = (q_1, ..., q_m)^{'}$ be two numerical vectors with $\min(a_1, ..., a_m) \geq 0, \max(a_1, ..., a_m) > 0$ and $0 = q_0 < q_1 < q_2 < ... < q_m = 1$. Then, for any continuous random variable $X$ and for any $p \in (0, 1)$ it holds that $$\text{VaR}_p[X_{A, Q}] = \text{VaR}_{p^{*}}[X]$$ and $$\text{TVaR}_p[X_{A, Q}] = 
\frac{a_{l}(1 - p^{*})}{a(1 - p)}\text{TVaR}_{p^{*}}[X] + \sum^{m - 1}_{k = l}\frac{(a_{k + 1} - a_{k})(1 - q_k)}{a(1 - p)}\text{TVaR}_{q_k}[X]$$ with $p^{*} = G^{-1}_{A, Q}(p) \in [q_{l - 1}, q_{l}]$ for some $l \in \{1, ..., m\}$ and $\sum_{k = m}^{m - 1} = 0$ by convention. 
\end{cor}
\begin{proof}
Corollary \ref{73264} follows immediately from Corollary \ref{1115119}. Indeed, for any fixed $p \in (0, 1)$, we consider the function $g$ defined by 
$$g(q) = \mathbb{I}(q > 1 - p), \,\, q \in [0, 1].$$ 
Then, it is easy to check that $g$ is a left-continuous distortion function with $$\rho_g[Y] = \text{VaR}_p[Y]$$ for any random variable $Y$, see \cite{DDGK} and \cite{Dhaene et al. (2012)} for the details. This combining Corollary \ref{1115119} implies that for any continuous random variable $X$ it holds that $$\text{VaR}_p[X_{A, Q}] = \text{VaR}_{p^{*}}[X], \,\, \text{for any} \,\, p \in (0, 1).$$ One can easily check that $p^{*} =G_{A, Q}^{-1}(p) \in (0, 1)$ for any $p \in (0, 1)$. 

For the TVaR of the step-weighted distribution at any probability level $p \in (0, 1)$, we consider the function $g$ defined by $$g(q) = \min\Big\{\frac{q}{1 - p}, 1\Big\}, \,\, q \in [0, 1].$$ Then, it is easy to check that $g$ is a continuous distortion function with $$\rho_g[Y] = \text{TVaR}_p[Y]$$ for any random variable $Y$, see \cite{DDGK} and \cite{Dhaene et al. (2012)} for instance. This combining Corollary \ref{1115119} implies that for any continuous random variable $X$ and any $p \in (0, 1)$ it holds that 
\begin{align*} 
\text{TVaR}_p[X_{A, Q}] &= \int^1_0F_X^{-1}(1 - q)dg_{A, Q}(q) \\ 
&= \frac{1}{1 - p}\int^{1}_{p^{*}}F_X^{-1}(t)dG_{A, Q}(t) \\ 
&= \frac{a_{l}(1 - p^{*})}{a(1 - p)}\text{TVaR}_{p^{*}}[X] + \sum^{m - 1}_{k = l}\frac{(a_{k + 1} - a_{k})(1 - q_k)}{a(1 - p)}\text{TVaR}_{q_k}[X]
\end{align*}
with $p^{*} = G^{-1}_{A, Q}(p) \in [q_{l - 1}, q_{l}]$ for some $l \in \{1, ..., m\}$ and $\sum_{k = m}^{m - 1} = 0$ by convention. This completes the proof. 
\end{proof}

To end this section, we present below an example in order to illustrate the main findings of this section. Moreover, this example will be useful when we consider comonotonic and moment matching approximations for sums of lognormal random variables. The interested reader is referred to \cite{Huang2026}.
\begin{ex} 
Let $A = (a_1, ..., a_m)^{'}$ and $Q = (q_1, ..., q_m)^{'}$ be two numerical vectors with $\min(a_1, ..., a_m) \geq 0, \max(a_1, ..., a_m) > 0$ and $0 = q_0 < q_1 < q_2 < ... < q_m = 1$. Let $X \sim LN(\mu, \sigma^2)$ be a log-normally distributed random variable with parameters $\mu \in \mathbb{R}$ and $\sigma > 0$ and let $X_{A, Q}$ be the step-weighted version of $X$ associated with the two numerical vectors $A$ and $Q$. Then, for any $p \in (0, 1)$, by Corollary \ref{73264} we have that $$\text{VaR}_p[X_{A, Q}] = e^{\mu + \sigma\Phi^{-1}(p^{*})}$$ and 
\begin{align*} 
\text{TVaR}_p[X_{A, Q}] 
&= \frac{e^{\mu + \sigma^2/2}}{a(1 - p)}\Bigg(a_{l}\Phi(\sigma - \Phi^{-1}(p^{*})) + \sum^{m - 1}_{k = l}(a_{k + 1} - a_{k})\Phi(\sigma - \Phi^{-1}(q_k))\Bigg)
\end{align*}
with $p^{*} = G^{-1}_{A, Q}(p) \in [q_{l - 1}, q_{l}]$ for some $l \in \{1, ..., m\}$ and $\sum_{k = m}^{m - 1} = 0$ by convention. 
\end{ex}

\section{Distortion risk measures of step-weighted comonotonic sums}\label{26538}
In this section, we investigate distortion risk measures of step-weighted comonotonic sums. Let $(X_1, ..., X_n)^{'}$ be a random vector of dimension $n$ and $\omega$ be a weight function. Then, we call $((X_1)_{\omega}, ..., (X_n)_{\omega})^{'}$ the weighted version of the random vector $(X_1, ..., X_n)^{'}$ associated with the weight function $\omega$. As usual, we denote the comonotonic sum of the components of $(X_1, ..., X_n)^{'}$ by $$S^c = X_1^c + ... + X_n^c$$ and the weighted version of $S^c$ associated with the weight function $\omega$ by $(S^c)_{\omega}$. The sum and the comonotonic sum of the components of the weighted random vector $((X_1)_{\omega}, ..., (X_n)_{\omega})^{'}$ will be denoted by 
$$S_{\omega} = (X_1)_{\omega} + ... + (X_n)_{\omega}$$ and $$(S_{\omega})^c = ((X_1)_{\omega})^c + ... + ((X_n)_{\omega})^c$$ respectively. Comonotonicity has been extensively discussed in \cite{Dhaene et al. (2002a)} and \cite{Dhaene et al. (2002b)}. For its applications in finance and actuarial science we recommend the reader to refer to \cite{DDV2011}, \cite{DVG2005}, \cite{DWY2000}, \cite{KDG2000}, \cite{VCD2008} and \cite{WD1998}. 

Then, we introduce the following fundamental result. For the proof see Theorem 5.1 of \cite{Huang2026}.
\begin{lem} \label{0103383}
Let $A = (a_1, ..., a_m)^{'}$ and $Q = (q_1, ..., q_m)^{'}$ be two numerical vectors with $\min(a_1, ..., a_m) \geq 0, \max(a_1, ..., a_m) > 0$ and $0 = q_0 < q_1 < q_2 < ... < q_m = 1$. Then for any continuous random variables $X_1, ..., X_n$ we have that $$(S^c)_{\omega} \stackrel{d}{=} (S_{\omega})^c$$
where $\stackrel{d}{=}$ stands for equality in distribution and $X_{\omega}$ denotes the step-weighted version of a continuous random $X$ associated with the step-weight function $\omega(x) = \omega_{A, Q, X}(x)$. 
\end{lem}
\noindent From Lemma \ref{0103383} we can conclude that the step-weighted version $(S^c)_{\omega}$ of the comonotonic sum $S^c$ of continuous random variables $X_1, ..., X_n$ is equal in distribution to the comonotonic sum $(S_{\omega})^c$ of the step-weighted versions of $X_1, ..., X_n$. In other words, the operation of comonotonic summing interchanges with the operation of step-weighting for continuous random variables. By Lemma \ref{0103383} we have the following result, which is useful when we want to compute distortion risk measures of the step-weighted comonotonic sum of continuous random variables.

\begin{thm} \label{73295}
Let $A = (a_1, ..., a_m)^{'}$ and $Q = (q_1, ..., q_m)^{'}$ be two numerical vectors with $\min(a_1, ..., a_m) \geq 0, \max(a_1, ..., a_m) > 0$ and $0 = q_0 < q_1 < q_2 < ... < q_m = 1$. Let $g$ be any fixed distortion function and define 
\begin{equation*}
g_{A, Q}(t) = g(\overline{G}_{A, Q}(t)) = g(1 - G_{A, Q}(1 - t)), \,\, \text{for} \,\, t \in [0, 1]
\end{equation*}
where $G_{A, Q}(t)$ is the continuous distortion function defined by (\ref{73366}). 
Then for any continuous random variables $X_1, ..., X_n$ we have that $$\rho_g[(S^c)_{A, Q}] = \sum^n_{i = 1}\rho_{g_{A, Q}}[X_i].$$
\end{thm}
\begin{proof} 
By Lemma \ref{0103383} and the well-known additivity of distortion risk measures for comonotonic sums we have that $$\rho_g[(S^c)_{A, Q}] = \rho_g[(S_{A, Q})^c] = \sum^n_{i = 1}\rho_g[(X_i)_{A, Q}].$$ From Theorem \ref{1230220} it follows that $$\rho_g[(S^c)_{A, Q}] = \sum^n_{i = 1}\rho_{g_{A, Q}}[X_i].$$ This completes the proof. 
\end{proof}
From Theorem \ref{73295} we see that associated with a distortion function $g$ the distortion risk measure $\rho_g[(S^c)_{A, Q}]$ of the step-weighted comonotonic sum $(S^c)_{A, Q}$ of any continuous random variables $X_1, ..., X_n$ can be expressed as the sum of the distortion risk measures of $X_1, ..., X_n$ associated with the distortion function $g_{A, Q}$.

\begin{cor} \label{73305}
Let $A = (a_1, ..., a_m)^{'}$ and $Q = (q_1, ..., q_m)^{'}$ be two numerical vectors with $\min(a_1, ..., a_m) \geq 0, \max(a_1, ..., a_m) > 0$ and $0 = q_0 < q_1 < q_2 < ... < q_m = 1$. Then, for any continuous random variables $X_1, ..., X_n$ and for any $p \in (0, 1)$ we have that $$\text{VaR}_p[(S^c)_{A, Q}] = \sum^n_{i = 1}\text{VaR}_{p^{*}}[X_i]$$ and $$\text{TVaR}_p[(S^c)_{A, Q}] = \sum^n_{i = 1}\Bigg(\frac{a_{l}(1 - p^{*})}{a(1 - p)}\text{TVaR}_{p^{*}}[X_i] + \sum^{m - 1}_{k = l}\frac{(a_{k + 1} - a_{k})(1 - q_k)}{a(1 - p)}\text{TVaR}_{q_k}[X_i]\Bigg)
$$ with $p^{*} = G^{-1}_{A, Q}(p) \in [q_{l - 1}, q_{l}]$ for some $l \in \{1, ..., m\}$ and $\sum_{k = m}^{m - 1} = 0$ by convention. 
\end{cor} 
\begin{proof} 
Corollary \ref{73305} follows immediately from Theorem \ref{73295}. Indeed, for any fixed $p \in (0, 1)$, we consider the function $g$ defined by 
$$g(q) = \mathbb{I}(q > 1 - p), \,\, q \in [0, 1].$$ 
Then, it is easy to check that $g$ is a left-continuous distortion function with $$\rho_g[Y] = \text{VaR}_p[Y]$$ for any random variable $Y$, see \cite{DDGK} and \cite{Dhaene et al. (2012)} for the details. This combining Theorem \ref{73295} implies that for any continuous random variables $X_1, ..., X_n$ it holds that 
\begin{align*} 
\text{VaR}_p[(S^c)_{A, Q}] &= \rho_g[(S^c)_{A, Q}] = \sum^n_{i = 1}\rho_{g_{A, Q}}[X_i] = \sum^n_{i = 1}\text{VaR}_{p^{*}}[X_i]. 
\end{align*}
One can easily check that $p^{*} =G_{A, Q}^{-1}(p) \in (0, 1)$ for any $p \in (0, 1)$.

For the TVaR of the step-weighted comonotonic sum $(S^c)_{A, Q}$ at any probability level $p \in (0, 1)$, we consider the function $g$ defined by $$g(q) = \min\Big\{\frac{q}{1 - p}, 1\Big\}, \,\, q \in [0, 1].$$ Then, it is easy to check that $g$ is a continuous distortion function with $$\rho_g[Y] = \text{TVaR}_p[Y]$$ for any random variable $Y$, see \cite{DDGK} and \cite{Dhaene et al. (2012)} for instance. This combining Theorem \ref{73295} implies that for any continuous random variables $X_1, ..., X_n$ it holds that 
\begin{align*} 
\text{TVaR}_p[(S^c)_{A, Q}] &= \rho_g[(S^c)_{A, Q}] = \sum^n_{i = 1}\rho_{g_{A, Q}}[X_i] \\ 
&= \sum^n_{i = 1}\Bigg(\frac{a_{l}(1 - p^{*})}{a(1 - p)}\text{TVaR}_{p^{*}}[X_i] + \sum^{m - 1}_{k = l}\frac{(a_{k + 1} - a_{k})(1 - q_k)}{a(1 - p)}\text{TVaR}_{q_k}[X_i]\Bigg)
\end{align*}
with $p^{*} = G^{-1}_{A, Q}(p) \in [q_{l - 1}, q_{l}]$ for some $l \in \{1, ..., m\}$ and $\sum_{k = m}^{m - 1} = 0$ by convention. This completes the proof. 
\end{proof}

To end this section, we present below two examples to illustrate the main findings of this section. We remark that the results of Example \ref{74351} will be useful when considering comonotonic and moment matching approximations for sums of lognormal random variables. The interested reader is referred to \cite{Huang2026}.

\begin{ex}
Let $A = (\frac{1}{2}, \frac{11}{6}, \frac{5}{6})^{'}$ and $Q = (\frac{1}{4}, \frac{1}{2}, 1)^{'}$. Let $X_1$ and $X_2$ be two random variables uniformly distributed over the unit interval $(0, 1)$. Consider the comonotonic sum $S^c = X_1^c + X_2^c$ where $(X_1^c, X_2^c)^{'}$ is a comonotonic modification of $(X_1, X_2)^{'}$. Let $(S^c)_{A, Q}$ be the step-weighted version of $S^c$ associated with the two numerical vectors $A$ and $Q$. By Corollary \ref{73305} we have the following table containing VaR's and TVaR's of the step-weighted comonotonic sum $(S^c)_{A, Q}$. 
\begin{figure} [!htbp] 
\centering 
\includegraphics[width = 0.75 \textwidth]{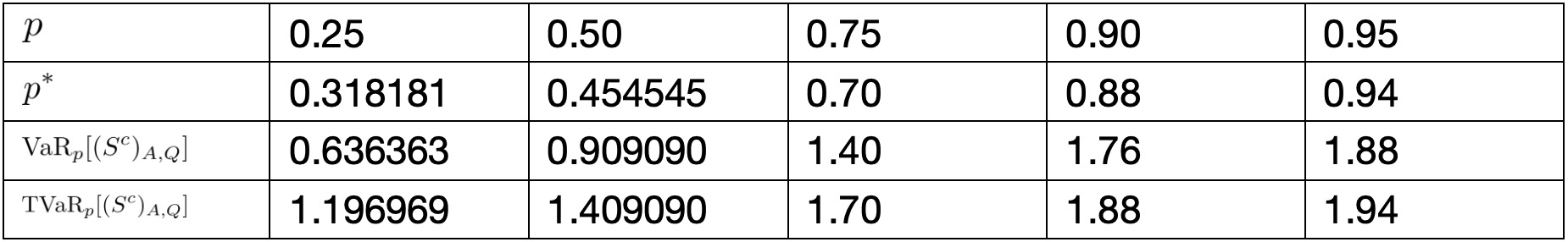} 
\caption*{Table 1: \text{VaR}'s and \text{TVaR}'s of $(S^c)_{A, Q}$.} 
\label{}
\end{figure}
One can directly check the table by the definition of $\text{VaR}$ and $\text{TVaR}$. Indeed, by the well-known additivity of quantiles for comonotonic sums we have that $S^c \sim U(0, 2)$, which implies that the step-weighted distribution of $S^c$ associated with the step-weight function $\omega = \omega_{A, Q}$ is given by $$F_{(S^c)_\omega}(x)\begin{cases} \frac{1}{4}x & x \in [0, \frac{1}{2}) \\ 
\frac{11}{12}(x - \frac{1}{2}) + \frac{1}{8} & x \in [\frac{1}{2}, 1) \\ 
\frac{5}{12}(x - 1) + \frac{7}{12} & x \in [1, 2]
\end{cases}.$$ From this, the table follows immediately. 
\end{ex}

\begin{ex} \label{74351}
Let $A = (a_1, ..., a_m)^{'}$ and $Q = (q_1, ..., q_m)^{'}$ be two numerical vectors with $\min(a_1, ..., a_m) \geq 0, \max(a_1, ..., a_m) > 0$ and $0 = q_0 < q_1 < q_2 < ... < q_m = 1$. Let $X_i \sim LN(\mu_i, \sigma_i^2)$ be a log-normally distributed random variable with parameters $\mu_i \in \mathbb{R}$ and $\sigma_i > 0$, $i = 1, ..., n$. Consider the comonotonic sum $S^c = X_1^c + ... + X_n^c$ of $X_1, ..., X_n$ and its step-weighted version $(S^c)_{A, Q}$ associated with the two numerical vectors $A$ and $Q$. Then, for any $p \in (0, 1)$, by Corollary \ref{73305} we have that 
$$\text{VaR}_p[(S^c)_{A, Q}] = \sum^n_{i = 1}e^{\mu_i + \sigma_i\Phi^{-1}(p^{*})}$$ and 
\begin{align*} 
\text{TVaR}_p[(S^c)_{A, Q}] 
&= \sum^n_{i = 1}\frac{e^{\mu_i + \sigma_i^2/2}}{a(1 - p)}\Bigg(a_{l}\Phi(\sigma_i - \Phi^{-1}(p^{*})) + \sum^{m - 1}_{k = l}(a_{k + 1} - a_{k})\Phi(\sigma_i - \Phi^{-1}(q_k))\Bigg)
\end{align*}
with $p^{*} = G^{-1}_{A, Q}(p) \in [q_{l - 1}, q_{l}]$ for some $l \in \{1, ..., m\}$ and $\sum_{k = m}^{m - 1} = 0$ by convention. These two expressions for the step-weighted comonotonic sum $(S^c)_{A, Q}$ will be useful when considering comonotonic and moment matching approximations for the sum of lognormal random variables. 
\end{ex}


\end{document}